%% file: Preemption.tex
\documentclass[manuscript, screen, nonacm]{acmart}
\usepackage{booktabs} 

\usepackage{cleveref}

\usepackage{lipsum}
\usepackage{amsfonts}
\usepackage{graphicx}
\usepackage{epstopdf}

\usepackage{algorithmicx}
\usepackage{algorithm}
\usepackage[noend]{algpseudocode}
\usepackage{cuted}
\usepackage{tikz}
\usepackage{tkz-euclide}
\usepackage{tkz-berge}
\usetikzlibrary{patterns}
\usetikzlibrary[patterns]
\usepgflibrary{patterns}
\usepgflibrary[patterns]

\begin{document}
\title{WSRPT is 1.2259-competitive for Weighted Completion Time Scheduling.}

\author{Samin Jamalabadi}
\affiliation{
  \institution{TU Dortmund University}
  \city{Dortmund}
  \country{Germany}
}\email{samin.jamalabadi@tu-dortmund.de}
\author{Uwe Schwiegelshohn}
\affiliation{
  \institution{TU Dortmund University}
  \city{Dortmund}
  \country{Germany}
}
\email{uwe.schwiegelshohn@tu-dortmund.de}

\renewcommand\shortauthors{Jamalabadi, S. and Schwiegelshohn, U.}

\begin{abstract}
\textit{Weighted shortest processing time first} (WSPT) is one of the best known algorithms for total weighted completion time scheduling problems. For each job $J_j$, it first combines the two independent job parameters weight $w_j$ and processing time $p_j$ by simply forming the so called Smith ratio $w_j/p_j$. Then it schedules the jobs in order of decreasing Smith ratio values. The algorithm guarantees an optimal schedule for a single machine and the approximation factor $1.2071$ for parallel identical machines.

For the corresponding online problem in a single machine environment with preemption, the \textit{weighted shortest remaining processing time first} (WSRPT) algorithm replaces the processing time $p_j$ with the remaining processing time $p_j(t)$ for every job that is only partially executed at time $t$ when determining the Smith ratio. Since more than 10 years, we only know that the competitive ratio of this algorithm is in the interval $[1.2157,2]$.

In this paper, we present the tight competitive ratio $1.2259$ for WSRPT. To this end, we iteratively reduce the instance space of the problem without affecting the worst case performance until we are able to analyze the remaining instances. This result makes WSRPT the best known algorithm for deterministic online total weighted completion time scheduling in a preemptive single machine environment improving the previous competitive ratio of $1.5651$. Additionally, we increase the lower bound of this competitive ratio from $1.0730$ to $1.1038$.  
\end{abstract}

%
%

\ccsdesc[500]{Theory of computation~Scheduling algorithms}
\ccsdesc[500]{Theory of computation~Online algorithms}
%
%

\keywords{Total weighted completion time, preemption, competitive ratio}

\maketitle

\input{Preemption-content}

\end{document}

%% file: Preemption-content.tex
\section{Introduction}
\label{sec:introduction}

The \textit{Shortest Processing Time first} (SPT) algorithm optimally solves the total completion time problem ($\sum C_j$) in a single machine environment ($1$) and a parallel identical machine  environment ($P$) by starting the unscheduled job $J_j$ with the shortest processing time $p_j>0$ whenever a machine is idle. In this problem, $C_j$ denotes the completion time of job $J_j$. If job $J_j$ has an additional and independent weight parameter $w_j\geq 0$, we can determine the Smith ratio $w_j/p_j$ for every job, see \cite{Smi56}, and start the unscheduled job with the largest ratio whenever a machine is idle. Due to its close relationship to SPT, we call this algorithm \textit{Weighted Shortest Processing Time first} (WSPT). The WSPT algorithm optimally solves the total weighted completion time problem ($\sum w_j\cdot C_j$) in a single machine environment and achieves the approximation factor $(1+\sqrt{2})/2$ in a parallel identical machine environment, see \cite{KaK86}. Variations of the WSPT approach are useful for many variants of this problem.

In this paper, we consider the problem variant in which each job $J_j$ has a release date $r_j\geq 0$ and there is only a single machine that allows preemptions, that is, we can interrupt the execution of a job and resume its execution at a later time without any change in the total processing time of the job. If all jobs have weight $1$ then the \textit{Shortest Remaining Processing Time first} (SRPT) algorithm optimally solves the total completion time problem, see \cite{Sch68}. The SRPT algorithm simply requires that at any time $t$, the machine always executes an available job with the shortest remaining processing time. A job $J_j$ is available at time $t$ if $r_j\leq t < C_j$ holds. We use $p_j(t)$ to denote the remaining processing time of job $J_j$ at time $t$. The result holds for the deterministic setting (all job parameters are available at the start of the schedule) and for the online setting (a job is unknown until its release date). In case of the online setting, we also say that jobs arrive over time. 

The corresponding completion time problem with weights is strongly NP-hard, see \cite {LaL84}. Since there is a PTAS for the deterministic problem, see \cite{ABC99}, the deterministic problem is settled. Therefore, we address the online problem with weights. Applying the previously introduced naming convention, we call this algorithm \textit{Weighted Shortest Remaining Processing Time first} (WSRPT). However, we only know that the competitive ratio of WSRPT is within the interval $[1.2157,2]$. Similarly, there is a significant gap for the competitive ratio of deterministic online weighted completion time scheduling in a single machine environment with preemption: the online algorithm with the best known competitive ratio $1.5651$ is due to \cite{Sit10b} while the best known lower bound of the competitive ratio is $1.0730$, see \cite{EpS03}.

\subsection{Our Contribution and Techniques}
\label{sec:results}

We show that WSRPT guarantees the competitive ratio $1.2259$ and that this competitive ratio is tight for WSRPT. Therefore, WSRPT is the best known algorithm for online total weighted completion time scheduling on a single machine with preemption. Our proof approach uses a restricted instance space with the property that we can map any problem instance onto a problem instance of the restricted instance space without improving the WSRPT performance of the original instance. To complete the proof, we determine the worst case instance in the restricted instance space. 

In particular, we show that it is sufficient to consider instances in which WSRPT does not provide a priority of jobs whenever a new job is submitted. In such an equality instance, we assume that WSRPT always selects the job that leads to the worst performance. Afterwards, we focus on a so called basic scenario in which the submission of many jobs with a very short processing time occur during the execution of a job with a long processing time. Finally, we show that the combination of two such scenarios produces the worst performance of WSRPT.  

Further, we use part of our analysis to improve the lower bound of the competitive ratio for this deterministic online problem from $1.0730$ to $1.1038$. The still existing gap is not surprising if we compare the worst case instance of WSRPT with the instance that establishes the lower bound of the competitive ratio. Although both instances use two jobs with long processing times and many jobs with very small processing times, the worst case instance of WSRPT has numerous different submission times and weight over processing time ratios while the lower bound instance has only two different submission times and weight over processing time ratios. Therefore, it seems likely that there are instances with multiple different submission times and weight over processing time ratios that may generate a larger lower bound of the competitive ratio. However, more submission times produce more decisions resulting in a significantly increased complexity of the analysis.   

\subsection{Related Work}
\label{sec:related} 

There is an optimal solution for the deterministic online total weighted completion time problem on a single machine without preemption: \cite{AnP02} show that the so called delayed-WSPT algorithm guarantees the competitive ratio $2$ while \cite{HoV96} prove that no deterministic algorithm can achieve a better bound even if all jobs have unit weight. We can adapt the WSPT algorithm to the preemptive environment by requiring that at any time, the machine executes the available and uncompleted job with the largest Smith ratio. This algorithm also achieves the competitive ratio $2$, see \cite{ScS02}. For WSRPT, \cite{Meg06} also proves the competitive ratio $2$ and gives an example showing that the bound of WSRPT can never be smaller than $1.2106$. Later \cite{XiC12} improve this lower bound to $1.2157$ resulting in the present performance gap for WSRPT.       

\cite{Sit10b} suggests a variant of the preemptive WSPT algorithm that only switches to a new job with a better Smith ratio if the remaining processing time of the running job is larger than a threshold. This approach combines preemptive WSPT and and some form of WSRPT since it considers both the original Smith ratio and the remaining processing time. It guarantees the currently best known competitive ratio $1.5651$. \cite{EpS03} provide different lower bounds of competitive ratios including the lower bound $1.0730$ for deterministic algorithms in a preemptive single machine environment.

There is some debate whether the total weighted completion time $\sum w_j\cdot C_j$ or the total weighted flow time $\sum w_j \cdot (C_j-r_j)$ is the most relevant completion time related objective function. For the problem without weights, SRPT produces the optimal result for completion time and flow time scheduling while the approximation ratio is significantly larger for total weighted flow time than for total weighted completion time. Recently, \cite{ARW23} settled the deterministic flow time problem by presenting a PTAS. Since there is no constant competitive algorithm for total weighted flow time on a single machine with preemption, see \cite{BaC09}, the exact competitive factor of WSRPT for this problem is not interesting. Therefore, online total weighted completion time scheduling on a single machine with preemption is the only basic open problem in this area and we focus on this problem without any statement about advantages and disadvantages of completion time and flow time objectives.

\section{WSRPT Algorithm}
\label{sec:WSRPT}

At every time $t$, the WSRPT algorithm runs the available job $J_j$ with the largest \textit{Smith ratio} $w_j/p_j(t)$. To this end, it may interrupt a running job and start another one. Note that in this paper, the \textit{Smith ratio} of Job $J_j$ uses the remaining processing time $p_j(t)$ while $w_j/p_j$ is the \textit{weight over processing time ratio} for $J_j$. As already mentioned, work on non-preemptive weighted completion time scheduling uses the expression Smith ratio for $w_j/p_j$.

We say that job $J_j$ is available from its submission time $r_j$ until its completion time $C_j$. Since the remaining processing time of the running job decreases with progressing time while the remaining processing time of all other jobs does not change, we must compare the Smith ratios of different jobs only at the submission time of a new job or at the completion time of a running job. If there are several jobs with the largest Smith ratio at such comparison then the algorithm arbitrarily selects one of these jobs. 

In this paper, we particularly consider \textit{equality instances} with decisions that involve several jobs with the same Smith ratio. Informally, WSRPT does not provide any benefit in an equality instance since it determines no priority among the involved jobs. We give the formal definition of such equality instance below:  
\begin{definition}
\label{def:eq_instance}
A problem instance is an equality instance if at every submission time, all jobs submitted at this time have the same Smith ratio and this ratio matches the Smith ratio of the running job if there is such job.
\end{definition}
\cite{Meg06} and \cite{XiC12} use equality instances as examples to show lower bounds of the competitive ratio of WSRPT. 

Since the expression competitive ratio always applies to all instances, we use the expression \textit{performance ratio} when we only consider a single instance or a subset of an instance. For this instance or one of its subsets, the performance ratio is the total weighted completion time of the (partial) WSRPT schedule over the total weighted completion time of the optimal schedule using the same jobs.
  
A makespan $C_{max}$ is the largest completion time in the WSRPT schedule of any job submitted before $C_{max}$. For our worst case analysis, we only consider WSRPT schedules that have a single makespan, that is, for such schedule, the makespan matches the common use of its name. If there is a submission of jobs at or after a $C_{max}$ value of an instance then we partition the instance into an instance containing all jobs submitted before $C_{max}$ and an instance containing all other jobs. 

In a WSRPT schedule with a single makespan there is no intermediate idle time but there may be an idle time interval preceding the schedule. If the WSRPT schedule of an instance has some idle time before the start time of the first job then we reduce the submission times of all jobs of this instance by the amount of this idle time. Since this transformation reduces the total completion times of the WSRPT schedule and the optimal schedule by the same amount, it increases the performance ratio. Therefore, we consider only instances with $0$ being the smallest submission time.  

To facilitate the description and the analysis of a WSRPT schedule, we introduce a segment of a WSRPT schedule: assume a job $J_j$ that starts at its submission time $r_j$ in the WSRPT schedule and a possibly empty set $\mathcal{J}_j$ of additional jobs with the same submission time $r_j$ and the same weight over processing time ratio as job $J_j$. The segment starts at this submission time $r_j$ and ends at the start time or continuation time of the first job $J_i$ that either has a smaller weight over processing time ratio than job $J_j$ or matches the weight over processing time ratio of $J_j$ but does not belong to set $\mathcal{J}_j$. If there is no such job then the segment ends with the makespan $C_{max}$ of the WSRPT schedule. Therefore, a segment of a WSRPT schedule is an interval of the schedule and there is no partial execution of a job in a segment. Segments are nested in a WSRPT schedule, that is, a segment cannot partially overlap with another segment.       
 
In our first lemma, we show that it is sufficient to only analyze equality instances.     
\begin{lemma}
\label{lem:reduction}
For any instance of the online total weighted completion time problem on a single machine with preemption, there is an equality instance with at least the same  performance ratio.
\end{lemma} 

\begin{proof}
Assume an arbitrary instance and its WSRPT schedule. The first transformation reduces the submission time of individual jobs while the second transformation reduces the weights of all jobs of a segment.

If the submission of a job $J_j$ occurs during the execution or at the start of a job $J_i$ and $w_j/p_j < w_i/p_i(r_j)$ holds then we reduce the submission time of $J_j$. To this end, we consider the largest time $t<r_j$ such that we have $w_j/p_j=w_k/p_k(t)$ for a job $J_k$ executing at time $t$ or starting its execution at time $t$ in the WSRPT schedule. If job $J_k$ has started its execution at time $t$ then we reduce $r_j$ to $r_k$ otherwise we reduce it to time $t$. If there is no such time $t$ then we reduce $r_j$ to $0$. Since WSRPT does not force job $J_j$ to start at any time after the new value and before the old value of $r_j$, we maintain our previous decisions and the WSRPT schedule does not change. The performance ratio of WSRPT for the instance cannot decrease since the reduction of submission times cannot increase the total weighted completion time of the optimal schedule. If  a job $J_j$ has a positive submission time $r_j>0$ and does not start at its submission time after this transformation then it has the same Smith ratio as the job running or starting at $r_j$ in the WSRPT schedule. 

Next we select a job $J_j$ that has a larger Smith ratio than the running job at its submission time $r_j$ and therefore interrupts the running job. Among these jobs, we pick the job with the largest weight over processing time. We collect all other jobs with the same submission time and the same Smith ratio into a set $\mathcal{J}_j$ and consider the segment of $J_j$ and set $\mathcal{J}_j$. Since we pick the job with the largest Smith ratio, all other jobs in the segment that start at their respective submission times have the same Smith ratio as the job they have interrupted. 

If this segment has at least the same performance ratio as the whole instance then we delete all jobs not belonging to this segment. This deletion cannot decrease the performance ratio of the segment since no job of the segment has a smaller submission time than $r_j$ and the deletion of additional jobs cannot increase the total weighted completion time of the segment jobs in the optimal schedule. 

If this segment has a smaller performance ratio than the whole instance then we scale down the weights of all jobs of the segment by the same factor until the candidate job has the same Smith ratio as the job with the same submission time and the second largest Smith ratio or if there is no such job and $r_j>0$ holds, it has the same Smith ratio as the job running at the submission time. Similar to the first transformation, this transformation does not require a change of the WSRPT schedule. Therefore, the transformation increases the performance ratio even if we do not change the optimal schedule. We proceed with this transformation until we obtain an equality instance.
\end{proof} 

We can apply a final transformation to an equality instance by replacing a job $J_j$ with $q$ identical jobs $J_i$ such that we have $p_i=p_j/q$, $w_i=w_j/q$, and $r_i=r_j$ and assume $q\rightarrow \infty$ provided there is no submission of any other job in the interval $(C_j-p_j,C_j)$ of the WSRPT schedule. This so called \textit{splitting transformation} reduces the total weighted completion times of the WSRPT schedule and of the optimal schedule by exactly $w_j\cdot p_j/2$ and by at least $w_j\cdot p_j/2$, respectively. Therefore, this transformation increases the performance ratio of the instance. Note that the improvement of the WSRPT lower bound $1.2157$ by \cite{XiC12} over the WSRPT lower bound $1.2106$ by \cite{Meg06} is due to this transformation. Due to the splitting transformation, there is at least one job submission during the processing of every job with a \textit{long} processing time in the WSRPT schedule while the processing of any \textit{small} job in the WSRPT schedule does not overlap with the submission of any other job. If a small job interrupts the processing of a long job in the WSRPT schedule, it has a larger weight over processing time ratio than the interrupted long job but the same Smith ratio since we only consider equality instances. Both jobs belong to the same segment.

\section{Analysis}
\label{sec:analysis}

If there is no long job then there is no preemption of a job in the WSRPT schedule and WSRPT generates an optimal schedule. Therefore. we distinguish several scenarios in the analysis based on long jobs.
\begin{description}
\item[Basic scenario] An instance with a single long job. The long job starts at its submission time in the WSRPT schedule.
\item[Nested long jobs] An instance with multiple long jobs. Every long job starts at its submission time in the WSRPT schedule.
\item[Delayed long jobs] An instance with multiple long jobs. A long job may not start at its submission time in the WSRPT schedule.
\end{description} 

In addition, we introduce a continuous function that we use twice during the analysis:
\begin{align*}
f(x) & = \frac{-(1-x)\cdot \ln (1-x)}{k\cdot x + c}  \mbox{ for } 0\leq x < 1, \mbox{ } c>0, \mbox{ and } k+c>0
\end{align*} 
The denominator of $f(x)$ is a positive linear function of $x$ for $0\leq x < 1$ while $-\ln (1-x) \cdot (1-x)$ is a concave function that has value $0$ for $x=0$, approaches $0$ for $x\rightarrow 1$ and has its maximum  at $x=1-1/e$. Therefore, $f(x)$ is also a concave function and has a single maximum at position $x_{max}$ with $0<x_{max}<1-1/e$ for $k>0$ and $1-1/e<x_{max}<1$ for $k<0$.     

\subsection{Basic Scenario}
\label{sec:basic1}
For long job $J_0$ of a basic scenario, we normalize weight and processing time $w_0=p_0=1$. As already discussed, we also assume $r_0=0$. In an equality instance, the preference of the running job or the submitted job does not affect the total completion time unless there is another submission during the execution of the selected job. Since there is no job submission during the execution of the small job by definition, the optimal schedule prefers a small job over the long job while the WSRPT schedule of the basic scenario generates a schedule with a worst case performance ratio by executing long job $J_0$ without preemption, that is, we have $C_0=p_0+r_0=1$. Moreover all $n$ small jobs have smaller submission times than $C_0$ since we only consider equality instances. 

For the sake of an easier description, we assume at the beginning that all small jobs have the same processing time $\delta$. Later it is beneficial to allow slightly different processing times. The property of an equality instance requires that the weight over processing time ratio $w_i/p_i$ of job $J_i$ for $1\leq i \leq n$ is identical to the Smith ratio of job $J_0$ at time $r_i$:
\begin{align*}
w_i & = \frac{p_i\cdot w_0}{p_0(r_i)} = \frac{p_i\cdot w_0}{p_0-r_i} = \frac{\delta}{1-r_i}
\end{align*}
We select the indexes of these jobs such that for $1 \leq i < j\leq n$, the submission time relation $r_i\leq r_j$ holds. Job $J_0$ and the set of all other jobs with submission time $r_0$ establish a segment that includes all jobs of the basic scenario. Since processing time $\delta$ is arbitrarily small and no small job interrupts any other small job, we assume that all submission times are multiples of $\delta$.

In the next few lemmas, we show that only basic scenarios with specific properties can produce the maximum performance ratio of basic scenarios, that is, we reduce the instance space for our analysis. For a specific instance of a basic scenario, $C$ and $C^*$ denote the total weighted completion time of the WSRPT schedule and the total weighted completion time of the optimal schedule, respectively. Similarly, $C_j$ and $C^*_j$ are the completion times of job $J_j$ in the WSRPT schedule and the optimal schedule, respectively.

\begin{lemma}
\label{lem:non-preempt}
For every instance of a basic scenario, there is always an instance of a basic scenario that executes job $J_0$ at the end of the optimal schedule without any preemption and has at least the same performance ratio.
\end{lemma}

We apply a proof by contradiction assuming that job $J_0$ occupies some early slot in the optimal schedule. Then we can remove this slot either by introducing a new small job or by exchanging one small job with the new one. One of both options increases the performance ratio.

\begin{proof}
If job $J_0$ does not execute at the end of the optimal schedule without preemption then there is a first job $J_j$ with submission time $r_j = i \cdot \delta > (j-1)\cdot \delta$  for $j\geq 1$, that is, the optimal schedule executes $J_0$ in interval $[(j-1)\cdot \delta, i\cdot \delta)$. We can remove this first execution slot of job $J_0$ in the optimal schedule by either introducing a new small job or by switching job $J_j$ against another small job.  

Assume the introduction of a new job $J'$ with submission time $r'= (i-1) \cdot \delta$, processing time $p'=\delta$, and weight $w'=\delta/(1-(i-1) \cdot \delta)$. This introduction removes the execution phase of $J_0$ in interval $[(i-1)\cdot \delta, i\cdot \delta)$ and increases $C$ and $C^*$ by $\Delta_I C$ and $\Delta_I C^*$, respectively: 
\begin{align*}
\frac{\Delta_I C}{\Delta_I C^*} & = 
\frac{(C_j+\delta)\cdot \frac{\delta}{1-(i-1)\cdot \delta} + \sum_{i=1}^{j-1}{w_i\cdot \delta}}{\frac{i\cdot \delta^2}{1-(i-1)\cdot \delta} +\delta} =
\frac{C_j+\delta+(1-(i-1)\cdot \delta )\cdot \sum_{i=1}^{j-1}{w_i}}{1+\delta} > \frac{C_j}{1+\delta}
\end{align*}
For  $C/C^* \leq  C_j/(1+\delta) < \Delta_I C/ \Delta_I C^*$, the introduction of $J'$ increases the performance ratio $C/C^*$.

For $C/C^* > C_j/(1+\delta)$, we remove job $J_j$ in addition to introducing job $J'$. This exchange of $J_j$ with $J'$ moves the execution phase of $J_0$ from interval $[(i-1)\cdot \delta, i\cdot \delta)$ to interval $[i\cdot \delta, (i+1)\cdot \delta)$ and reduces $C$ and $C^*$ by $\Delta_E C$ and $\Delta_E C^*$, respectively:   
\begin{align*}
\frac{\Delta_E C}{\Delta_E C^*} & = 
\frac{C_j\cdot \left(\frac{\delta}{1-i\cdot \delta} - \frac{\delta}{1-(i-1)\cdot \delta}\right)}{\frac{(i+1)\cdot \delta^2}{1-i\cdot \delta} - \frac{i\cdot \delta^2}{1-(i-1)\cdot \delta}} 
= \frac{C_j}{1+\delta} 
\end{align*}
Therefore, it increases the performance ratio $C/C^*$. A repeated application of both procedures depending on value of the performance ratio yields the claim.
\end{proof}	

Due to Lemma~\ref{lem:non-preempt}, the WSRPT schedule and the optimal schedule of a basic scenario with the largest performance ratio are non-preemptive schedules. 

In the next lemma, we discuss the average completion time of the jobs with the maximal weight over processing time ratio in the WSRPT schedule and in the optimal schedule.

\begin{lemma}
\label{lem:max_ratio}
For every basic scenario, there is a basic scenario with at least the same performance ratio such that this performance ratio does not exceed the ratio of the average completion time of all jobs with the maximal weight over processing time ratio $1/(1-y)$ in the WSRPT schedule over the average completion time of the same jobs in the optimal schedule.
\end{lemma}

\begin{proof}
Let job $J_j$ be the first job with the maximal weight over processing time ratio $1/(1-y)$ in an instance of a basic scenario. All jobs $J_j$ to $J_n$ are identical with submission time $y$. The WSRPT schedule and the optimal schedule execute these jobs one after the other with $C_j=1+\delta$ in the WSRPT schedule and $C_j^*=y+\delta$ in the optimal schedule. Then $C_j+(n-j)\cdot \delta/2$ and $C_j^*+(n-j)\cdot \delta/2$ are the average completion times of these jobs in the WSRPT schedule and the optimal schedule, respectively. 

For $C/C^*> (C_j+(n-j)\cdot \delta/2) / (C_j^*+(n-j)\cdot \delta/2)$, we reduce the submission time and the weight of jobs $J_{j}$ to $J_n$ to $y-\delta$ and $\delta /(1-(y-\delta))$, respectively. The reduction increases the performance ratio and decreases $C_j^*$ while $C_j$ remains unchanged after adjusting the job index $j$ to the new maximal weight over processing time ratio. We apply this reduction until the claim holds. 
\end{proof}	

We only use Lemma~\ref{lem:max_ratio} to show that a small job has its largest possible weight if it starts at its submission time in the optimal schedule. Later we will prove in Lemma~\ref{lem:group}, that there is only one job with the maximal weight over processing time ratio in a basic scenario with the worst performance ratio. 

\begin{lemma}
\label{lem:max_weight}
Let job $J_j$ be the first job with the maximal weight over processing time ratio $1/(1-y)$ in an instance of a basic scenario. For every such basic scenario, there is a basic scenario that has at least the same performance ratio and completes a job with weight $\delta / (1-(i-1)\cdot \delta)$ at time $i\cdot \delta \leq C_j^* =y+\delta$ in the optimal schedule for every $i \leq y/\delta$.
\end{lemma}

\begin{proof}
We consider a job $J_h$ with $w_h<\delta/(1-(i-1)\cdot \delta)$ and $C_h^*=i\cdot \delta<y+\delta =C_j^*$ and assume that Lemma~\ref{lem:max_ratio} holds. Due to $C_h>C_j+(n-j)\cdot \delta$, increasing the weight of $J_h$ also increases the performance ratio.
\end{proof}
Due to Lemmas~\ref{lem:non-preempt} and \ref{lem:max_weight}, we address basic scenarios that only execute jobs at their submission times in interval $[0,y)$ of the optimal schedule. Then the total contribution of the jobs scheduled in interval $[0,y)$ in the optimal schedule is 
\begin{align*}
\lim_{\delta\rightarrow 0} \sum_{i=1}^{\frac{y}{\delta}} \frac{i\cdot \delta}{1-(i-1)\cdot \delta} & = \int_0^y \frac{x}{1-x} dx = -y - \ln(1-y).
\end{align*}
In such basic scenario, a small job contributes more than the performance ratio to this performance ratio if it has less than the maximum weight over processing time ratio and starts at its submission time in the optimal schedule. The basic scenarios of the lower bound schedules presented by \cite{Meg06} and \cite{XiC12} observe the properties of Lemmas~\ref{lem:non-preempt}, \ref{lem:max_ratio}, and \ref{lem:max_weight}. However, these examples assume a given total length of jobs with the maximum weight over processing time ratio instead of optimizing this total length.  

The next lemma shows that a basic scenario with the worst performance ratio not only requires this optimization but also has a different structure. Therefore, this lemma is one of the key lemmas of this paper. In this lemma, we address the remaining interval of the optimal schedule. This interval contains all jobs that complete later than time $y$ in the optimal schedule. We partition these small jobs into groups of jobs with the same weight over processing time ratio, that is, with jobs submitted at the same time. Since the WSRPT schedule and the optimal schedule are non-preemptive, all jobs of a group execute one after the other in the optimal schedule and in the WSRPT schedule. For the sake of simplicity, we assume at the moment that all jobs of a group are identical. Since one small job completes at time $y+\delta$ in the optimal schedule, there is at least one group in every basic scenario, see Lemma~\ref{lem:max_ratio}. We use $q_i$ to indicate the number of jobs with weight $\delta/(1-i_i\cdot \delta)$ in group $i$ of a basic scenario. The last group $g$ contributes $C_g$ and $C^*_g$ to $C$ and $C^*$, respectively: 
\begin{align}
\label{eq:group_WSRPT}
C_g & = 
\left(1+y+\left(-i_g+\sum_{i=1}^{g-1} q_i +  \frac{\left\lceil q_g\right\rceil}{2} \right) \cdot \delta\right)\cdot \frac{q_g\cdot \delta}{1-i_g\cdot \delta} + \sum_{i=1}^{i_g} \frac{\delta}{1-(i-1)\cdot \delta}\cdot q_g\cdot \delta \\
\label{eq:group_optimal}
C^*_g & = \left(y+\left(\sum_{i=1}^{g-1} q_i + \frac{ \left\lceil q_g\right\rceil}{2} \right)\cdot \delta \right) \cdot \frac{q_g\cdot \delta}{1-i_g\cdot \delta} + q_g\cdot \delta
\end{align}
The average completion time of the jobs in group $g$ is $\lceil q_g\rceil/2 \cdot \delta$. Using Eq.~\eqref{eq:group_WSRPT} and \eqref{eq:group_optimal}, we obtain
\begin{align} 
\label{eq:group_ratio}
t_g (i_g \cdot \delta) & =\lim_{\delta \rightarrow 0} \frac{C_g}{C^*_g} = \frac{\frac{1+y+\left(-i_g+\sum_{i=1}^{g-1} q_i + \frac{\left\lceil q_g\right\rceil}{2}\right) \cdot \delta}{1-i_g\cdot \delta}-\ln (1-i_g\cdot \delta)}{\frac{y+\left(\sum_{i=1}^{g-1} q_i + \frac{\left\lceil q_g\right\rceil}{2}\right)\cdot \delta}{1-i_g\cdot \delta}+1} 
= 1 - \frac{\ln (1-i_g\cdot \delta)\cdot (1-i_g\cdot \delta)}{1+y+\left(-i_g+\sum_{i=1}^{g-1} q_i +\frac{ \left\lceil q_g\right\rceil}{2}\right) \cdot \delta}.
\end{align}
Equation~\eqref{eq:group_ratio} contains function f(x) introduced in the beginning of Section~\ref{sec:analysis} and provides the base for the key lemma.

\begin{lemma}
\label{lem:group}
Let $1/(1-y)$ be the largest weight over processing time ratio in a basic scenario. For every basic scenario, there is a basic scenario with at least the same performance ratio and the following properties:
\begin{enumerate}
\item For the last group $g$ of the basic scenario, $t_g(i_g\cdot \delta)$ is the performance ratio.
\item There is a group containing a single job for every $i_{g'}$ with $i_g \leq i_{g'} \leq y/\delta$. This job has a slightly larger processing time than $\delta$.
\item For every group, the ratio of the completion time of the single job in this group in the WSRPT schedule over the completion time of this job in the optimal schedule is the performance ratio. 
\end{enumerate}
\end{lemma}

\begin{proof}
We describe several modifications that we may apply repeatedly and/or in some combination. First, we discuss function $t_g(x)$ in Eq.~\eqref{eq:group_ratio}. Function $t_g(x)-1$ is function $f(x)$ with $k<0$ introduced in Section~\ref{sec:analysis}.

For $t_g(i_g \cdot \delta) < C/C^*$, we remove group $g$. Such removal increases the performance ratio of the basic scenario. 

For $t_{g}(i_{g}\cdot \delta)> C/C^*$, we introduce a new last group $g+1$ with $i_{g+1} < i_g$ and a suitable $q_{g+1}$ such that $t_{g+1}(i_{g+1}\cdot \delta)=C/C^*$ holds. This introduction does not decrease the performance ratio. Therefore, we have $t_g(i_g\cdot \delta) = C/C^*$.

We consider some group $h$ and determine ratio $r_h$:
\begin{align*}
r_h & = \frac{1+y +\left(-i_h+\sum_{i=1}^{h-1} q_i + \frac{\left\lceil q_h\right\rceil}{2}\right) \cdot \delta}{y+\left(\sum_{i=1}^{h-1} q_i + \frac{\left\lceil q_h\right\rceil}{2}\right)\cdot \delta}
\end{align*}
For $r_h > C/C^*$ with $h>1$, we increase $i_h$ resulting in a decrease of $r_h$ and an increase of the performance ratio. For $h=2$, the increase may change $y$.

For $r_h < C/C^*$, a decrease of $i_h$ produces increases of $r_h$ and of the performance ratio. Therefore, we assume that $r_h=C/C^*$ holds for every group $h$.

For $q_h>1$, we split group $h$ into group $h'$ followed by group $h''$ with $i_{h'}=i_{h''}=i_h$ and $q_{h'}+q_{h''}=q_h$. Therefore, we have $r_{h'} > C/C^* > r_{h''}$ and modify $i_{h'}$ and $i_{h''}$ as explained above. Since this procedure increases the performance ratio, repeated application of this procedure produces groups of minimal size. To guarantee the relationship $r_h=C/C^*$, we allow processing times slightly larger than $\delta$ for jobs completing after $y$ in the optimal schedule. Then every group consists of a single job.
\end{proof}
We only consider basic scenarios that observe Lemmas~\ref{lem:non-preempt}, \ref{lem:max_weight}, and \ref{lem:group}. In such basic scenario, we have $C_j/C_j^*=C/C^*$ for every small job $J_j$ that does not start at its submission time in the optimal schedule. Informally, these jobs do not directly increase the performance ratio but increase the ratio $C_k/C_k^*$ of a small job $J_k$ that has a smaller weight over processing time ratio than job $J_j$ and starts at its submission time in the optimal schedule. Only the long job in such basic scenario reduces the performance ratio.

We use $v=i_g \cdot \delta$ with $g$ being the last group in such basic scenario. Since the processing times of the small jobs are arbitrarily small, we apply a continuous extension of the problem: for any value $x \in [v,y]$, let $\Delta_x$ be the total (continuous) amount of jobs with at least $1/(1-x)$ as weight over performance ratio. For all $x \in [v,y]$, the worst case requires
\begin{align*}
\frac{1+ \Delta_x}{x+ \Delta_x} = \frac{C}{C^*} = \frac{1}{y} & \Leftrightarrow
\Delta_x  = \frac{y-x}{1-y} 
\end{align*}
resulting in 
\begin{align*}
-\ln (1-v) = \frac{C}{C^*} & = 1 - \frac{\ln (1-v) \cdot (1-v)}{1+ \Delta_v}= 1 - \ln (1-v) \cdot (1-y)
\end{align*}
due to Eq.~\eqref{eq:group_ratio} and Lemma~\ref{lem:group}. 

Lemma~\ref{lem:group} yields the worst case performance ratio of any basic scenario. Note that increasing the weight of any small job $J_j$ with less than the maximum weight over processing time ratio is impossible for jobs starting in interval $[0,y)$ in the optimal schedule due to the equality instance property while it reduces the ratio $C_j/C_j^*$ below the performance ratio for all other small jobs. Due to Lemma~\ref{lem:non-preempt}, we must execute the long job at the end of the basic scenario with worst case performance. 

\begin{lemma}
\label{lem:basic1}
The performance ratio of a basic scenario is at most $R_B=1.2259$.
\end{lemma} 
 
\begin{proof}
Using the continuous extension, the worst case expression of a basic scenario 
\begin{align*}
\frac{C}{C^*} & = \frac{1 + \int_0^{\Delta_v}\frac{\Delta_x+1}{1-x}d\Delta_x + \int_0^v \frac{1+\Delta_v+v-x}{1-x}dx}{\int_0^y \frac{x}{1-x}dx + \int_y^{\Delta_v+v} \frac{\Delta_x+x}{1-x}d(\Delta_x+x)+v+\Delta_v+1} \\
& = \frac{1+\int_0^{\Delta_v}\frac{1}{1-y}d\Delta_x + \int_0^v (1+\frac{\Delta_v+v}{1-x})dx }{\int_0^y (-1+\frac{1}{1-x})dx + \int_y^{\Delta_v+v} \frac{y}{1-y}d(\Delta_x+x)+v+\Delta_v+1} = \frac{1+\frac{y-v}{(1-y)^2}+v-\frac{y\cdot (1-v)}{1-y}\cdot \ln (1-v)}{-y-\ln (1-y) + \frac{y^3-v\cdot y^2}{(1-y)^2}+\frac{1-v\cdot y}{1-y}}
\end{align*}
yields the performance ratio $1.2259$ for $y=0.8157$ and $v=0.7066$ by applying numerical optimization. 
\end{proof}
In Fig~\ref{fig:basic1}, we show Gantt charts for the WSRPT schedule and the optimal schedule of the basic scenario with the largest performance ratio. The Gantt charts use large values for $\delta$ to better picture the use of these jobs. Note that the schedule of the jobs is only an informal representation showing the long job, the small jobs starting not later than $y$ in the optimal schedule and the slightly larger small jobs representing the groups. On top of the Gantt charts, we display a continuous weight profile graph of the corresponding schedule. At time $t$, the weight profile is the weight over processing time ratio for the job executing at this time in the represented schedule. Therefore, the continuous weight profile of the long job $J_0$ is a horizontal line with value $1$ for the time interval in which the schedule executes $J_0$ while for every small job, the continuous weight profile consists of a single point at the time of the execution of this job.
\begin{figure}[ht]
\centering
\includegraphics[width=9cm]{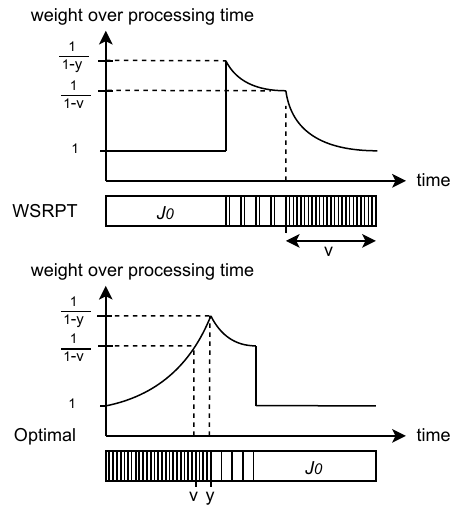}
\caption{\label{fig:basic1} Basic scenario: WSRPT and optimal schedules with their continuous weight profiles on top}
\end{figure}
Since the total weight $W$ and the length $L$ of a segment are important in the next chapter, we introduce their expressions for the basic scenario with the worst performance ratio.
\begin{align}
\label{eq:weight}
W & =1 + \int_0^{\Delta_v} \frac{1}{1-y\cdot (1+\Delta_x)+\Delta_x}d\Delta_x + \int_0^v \frac{1}{1-x}dx = 1+\frac{1}{1-y}\cdot \ln \left( \frac{1-v}{1-y}\right) - \ln (1-v)= 4.7521 \\
\label{eq:length}
L & = 1+ \Delta_v + v = \frac{1-v\cdot y}{1-y}=2.2995
\end{align}

In Table~\ref{tab:basic_result}, we show the numerically derived results for $z$, $L$, $W$, $C$, $C^*$, and $C/C^*$ in dependence of $y$. The variable $z$ denotes the size of a single group with the largest weight over processing time ratio.
\begin{table}
\begin{tabular}{|c|c|c|c|c|c|c|c|c|}
$y$ & $v$ & $z$ & $C$ & $C^*$ & $C/C^*$ & $W$ & $L$ & $W/L$ \\ \hline
$0.10$ & - & - & $1.1105$ & $1.1054$ & $1.0047$ & $1.1054$ & $1.1000$ & $1.0049$ \\ \hline
$0.20$ & - & - & $1.2446$ & $1.2231$ & $1.0176$ & $1.2231$ & $1.2000$ & $1.0193$ \\ \hline
$0.30$ & - & - & $1.4070$ & $1.3567$ & $1.0371$ & $1.3567$ & $1.3000$ & $1.0436$ \\ \hline
$0.10$ & - & $1.3270$ & $3.7031$ & $3.5581$ & $1.0407$ & $2.5798$ & $2.4270$ & $1.0630$ \\ \hline
$0.40$ & - & - & $1.6043$ & $1.5108$ & $1.0619$ & $1.5108$ & $1.4000$ & $1.0792$ \\ \hline
$0.20$ & - & $1.2335$ & $4.0187$ & $3.7216$ & $1.0799$ & $2.7675$ & $2.4355$ & $1.1363$ \\ \hline
$0.50$ & - & - & $1.8466$ & $1.6931$ & $1.0906$ & $1.6931$ & $1.5000$ & $1.1288$ \\ \hline
$0.30$ & - & $1.1384$ & $4.3650$ & $3.9086$ & $1.1168$ & $2.9830$ & $2.4384$ & $1.2233$ \\ \hline
$0.60$ & - & - & $2.1498$ & $1.9163$ & $1.1218$ & $1.9163$ & $1.6000$ & $1.1977$ \\ \hline
$0.40$ & - & $1.0337$ & $4.7457$ & $4.1241$ & $1.1507$ & $3.2337$ & $2.4337$ & $1.3287$ \\ \hline
$0.70$ & - & - & $2.5428$ & $2.2040$ & $1.1537$ & $2.2040$ & $1.7000$ & $1.2965$ \\ \hline
$0.50$ & - & $0.9186$ & $5.1643$ & $4.3742$ & $1.1806$ & $3.5303$ & $2.4186$ & $1.4597$ \\ \hline
$0.80$ & - & - & $3.0876$ & $2.6094$ & $1.1832$ & $2.6094$ & $1.8000$ & $1.4497$ \\ \hline
$0.90$ & - & - & $3.9723$ & $3.3026$ & $1.2028$ & $3.3026$ & $1.9000$ & $1.7382$ \\ \hline
$0.92$ & - & - & $4.2437$ & $3.5257$ & $1.2036$ & $3.5257$ & $1.9200$ & $1.8960$ \\ \hline
$0.60$ & - & $0.7884$ & $5.6201$ & $4.6643$ & $1.2049$ & $3.8873$ & $2.3884$ & $1.6276$ \\ \hline
$0.70$ & - & $0.6344$ & $6.0920$ & $4.9894$ & $1.2210$ & $4.3186$ & $2.3344$ & $1.8500$ \\ \hline
$0.71$ & $0.7043$ & $0.5922$ & $6.1372$ & $5.0223$ & $1.2220$ & $4.3656$ & $2.3273$ & $1.8758$ \\ \hline
$0.75$ & $0.7062$ & $0.3623$ & $6.3196$ & $5.1599$ & $1.2247$ & $4.5538$ & $2.3072$ & $1.9737$ \\ \hline
$0.76$ & $0.7063$ & $0.3059$ & $6.3639$ & $5.1944$ & $1.2252$ & $4.5985$ & $2.3044$ & $1.9955$ \\ \hline
$0.77$ & $0.7064$ & $0.2485$ & $6.4055$ & $5.2270$ & $1.2255$ & $4.6404$ & $2.3023$ & $2.0156$ \\ \hline
$0.78$ & $0.7064$ & $0.1949$ & $6.4443$ & $5.2578$ & $1.2257$ & $4.6789$ & $2.3010$ & $2.0334$ \\ \hline
$0.79$ & $0.7065$ & $0.1401$ & $6.4751$ & $5.2823$ & $1.2258$ & $4.7105$ & $2.2999$ & $2.0481$ \\ \hline
$0.80$ & $0.7065$ & $0.0855$ & $6.4996$ & $5.3020$ & $1.2259$ & $4.7352$ & $2.2995$ & $2.0592$ \\ \hline
$0.81$ & $0.7065$ & $0.0312$ & $6.5149$ & $5.3154$ & $1.2259$ & $4.7502$ & $2.2994$ & $2.0658$ \\ \hline
$0.8157$ & $0.7066$ & - & $6.5168$ & $5.3160$ & $1.2259$ & $4.7521$ & $2.2995$ & $2.0666$ \\ \hline
\end{tabular}

\vspace{10pt}

\caption{\label{tab:basic_result} Results for various values of $y$, $v$, and $z$}
\end{table}

\subsection{Multiple Long Jobs}
\label{sec:multiple}

We extend the basic scenario by allowing multiple long jobs. First we require that every long job starts at its submission time. If long job $J_j$ interrupts long job $J_k$ then $r_k< r_j$ and $C_j > C_k$ holds. Therefore, we use the name nested long jobs.

As already mentioned, each long job establishes a segment if it starts at its submission time. We combine a segment $S$ and a basic scenario $B$ with the corresponding long jobs $J_s$ and $J_b$. Using our previous notation, $S$ and $B$ have total weighted completion times $C_S$ and $C_B$ in the WSRPT schedule, total weighted completion times $C_S^*$ and $C_B^*$ in the optimal schedule, weights $W_S$ and $W_B$, and lengths $L_S$ and $L_B$, respectively. All these values assume normalization of the corresponding long jobs, that is, the relations $p_s=p_b=1$, $w_s=w_s=1$, and $r_s=r_b=0$ hold for the long jobs. Therefore, we must apply some correction depending on the actual values of the long jobs when using these values in a scenario. Furthermore, the parameters for basic scenario $B$ ignore the impact of segment $S$.  We normalize the first long job $J_b$ resulting in $r_{b}=0$ and $p_{b}=w_{b}=1$, see also Section~\ref{sec:basic1}. 

First we assume that segment $S$ is also a basic scenario.

\begin{lemma}
\label{lem:two_bs}
The nesting of a basic scenario within another basic scenario may produce a larger performance ratio than $R_B$.
\end{lemma}

\begin{proof}
Let $S$ be a basic scenario with the worst performance ratio $R_B$. We use $r_s=0.5307$ and assume $r_s+p_s\cdot L_S > v + \Delta_v = y\cdot (1-v)/(1-y)$, see Section~\ref{sec:basic1} for the definition of $y$ and $v$. Due to Lemma~\ref{lem:group}, no small job 
executes after basic scenario $S$ in the optimal schedule. Then we have  
\begin{align}
\label{eq:combined_ratio}
\frac{C}{C^*} & = \frac{1+r_s\cdot (1- \ln (1-r_s))+ w_s \cdot p_s\cdot C_S+ r_s\cdot w_s\cdot W_S + p_s\cdot L_S\cdot (1-\ln (1-r_s))}{1-\ln (1-r_s)+w_s \cdot p_s\cdot C_S^*+ r_s\cdot w_s\cdot W_S+ p_s\cdot L_S}.
\end{align}

This nesting of two basic scenarios has a larger performance ratio than $R_B$ for 
\begin{align*}
R_B & <  \frac{1+r_s\cdot (1- \ln (1-r_s))+ r_s\cdot w_s\cdot W_S + p_s\cdot L_S\cdot (1-\ln (1-r_s))}{(1-\ln (1-r_s))+r_s\cdot w_s\cdot W_S+ p_s\cdot L_S} =  1 + \frac{\frac{1-(1-r_s)\cdot (1- \ln (1-r_s))}{p_s\cdot L_S}- \ln (1-r_s)}{
r_s\cdot \frac{w_s}{p_s}\cdot \frac{W_S}{L_S}+ \frac{1- \ln (1-r_s)}{p_s\cdot L_S}+1} 
\\
& = 1 + \frac{(1-r_s) \cdot \left(\frac{r_s}{p_s\cdot L_S}-\left(\frac{1-r_s}{p_s\cdot L_S}+1 \right)\cdot \ln (1-r_s) \right)}{r_s\cdot \frac{W_S}{L_S}+ (1-r_s)\cdot \left(\frac{1-\ln (1-r_s)}{p_s\cdot L_S} + 1\right)}.
\end{align*}
The inequality holds for $p_s > 35$.
\end{proof}
We obtain the instance with the worst performance ratio by applying Lemma~\ref{lem:group} to all small jobs executing after $y$ in the optimal schedule. An optimization produces values for $r_s$, $p_s$, $y$, and $v$. While we obtain $r_s=0.5307$, we cannot detect a numerical change for $y$ and $v$ compared to Lemma~\ref{lem:basic1}. Therefore, the worst performance ratio is still $1.2259$. Since the performance ratio does not change, there is a wide numerical range of possible values for $p_s$ although the evaluation of Eq.~\eqref{eq:combined_ratio} in the proof of Lemma~\ref{lem:two_bs} yields exactly one optimal value of $p_s$ if all other parameters are fixed. We use $R$ to denote the worst performance ratio obtained from these two basic scenarios although $R$ and $R_B$ have the same numerical representation with the used accuracy.   

Next, we show that no WSRPT schedule can exceed the performance ratio $R$. To this end, we apply induction in the number of long jobs. Lemma~\ref{lem:basic1} and \ref{lem:two_bs} provide the induction base. We use the restriction of the instance space as established in Sections~\ref{sec:WSRPT} and \ref{sec:basic1}.

\begin{lemma}
\label{lem:nested_jobs}
The performance of a WSRPT schedule cannot exceed $R$ if every long job starts at its submission time.
\end{lemma}
The proof of Lemma~\ref{lem:nested_jobs} uses concepts that we have applied in other proofs.  

\begin{proof}
We consider the WSRPT schedule of a basic scenario $B$ and introduce a segment $S$. The performance ratio $C_S/C_S^*$ of this segment is at most $R$. We use the notations described at the beginning of this section. 

The claim holds for a concatenation of $S$ after $B$ since the delay of $S$ causes a reduction of the performance ratio of $S$ without affecting the performance ratio of $B$.   
 
Therefore, we assume that segment $S$ executes inside basic scenario $B$, that is, long job $J_s$ interrupts long job $J_b$. Lemmas ~\ref{lem:max_weight} and \ref{lem:group} hold for the whole WSRPT schedule while there may be preemption of long jobs in the optimal schedule.

We show that it is sufficient to consider basic scenarios $B$ without any submission of jobs after segment $S$. A larger performance ratio than $R$ for $C/C^*$ requires 
\begin{align}
\label{eq:cond}
R & < 1 - \frac{(1-r_s)\cdot \ln (1-r_s)}{r_s\cdot \frac{W_S}{L_S}+ (1-r_s)\cdot (1+W_{B+})},
\end{align}
see the proof of Lemma~\ref{lem:two_bs} if we ignore the contribution of $C_B$ and $C_B^*$ and assume that there is a submission of some jobs of basic scenario $B$ after segment $S$ with total weight $W_{B+}$. Then increasing the submission time $r_s$ of long job $J_s$ by $\delta$ results in an increase of the right hand side of Eq.~\eqref{eq:cond} since we have
\begin{align*}
 \frac{(r_s+\delta)\cdot  p_s \cdot W_S}{1-r_s-\delta} - \frac{r_s\cdot  p_s \cdot W_S}{1-r_s} -\frac{\delta\cdot p_s\cdot L_S}{1-r_s-\delta} & = \frac{\delta\cdot p_s}{1-r_s-\delta)}\cdot \left( \frac{W_S}{1-r_s}- L_S \right) >0.
\end{align*} 
Therefore, we can assume that there is no small job submission after the segment.
  
We assume that $p_s\cdot L_S \leq (1-R\cdot r_s)/(R-1)= \Delta_{r_s}$ and $C_S/C_S^*=R$ hold. We compare this combination of $S$ and $B$ with the worst case basic scenario, see Lemma~\ref{lem:group}. The worst case basic scenario delays each small job of $B$ by at least the same amount and its average contribution of its additional small jobs starting in interval $[x, \Delta_x)$ of the optimal schedule is larger than $R$. Therefore, this combination of $S$ and $B$ cannot exceed performance ratio $R$. 

Next, we consider $p_s\cdot L_S > (1-R\cdot r_s)/(R-1)$. For $r_s+p_s\cdot L_S < v +\Delta_v = y\cdot (1-v)/(1-y)$, we extend $p_s$ until $r_s+p_s\cdot L_S = v +\Delta_v = y\cdot (1-v)/(1-y)$ holds and remove all small jobs of basic scenario $B$ that execute after segment $S$ in the optimal schedule. Since all removed jobs contribute $R$ to the performance ratio, the removal has no impact. The increase of $p_s$ can delay small jobs with at most a total weight of $-(\ln (1-r_s))+\ln (1-v)< -\ln (1-r_s)$. Therefore, the contribution is less than the contribution for a further extension of $p_s$. Therefore, we assume $r_s+p_s\cdot L_S \geq y\cdot (1-v)/(1-y)$. Since no small job completes after the segment in the optimal schedule, we have $C_B/C_B^*=r_s +1/(1-\ln (1-r_s))$, see the proof of Lemma~\ref{lem:two_bs}. Since $W_{B+}=0$ holds,  Eq.~\eqref{eq:cond} includes function $f(x)$ from Section~\ref{sec:analysis}. The analysis of this function yields that the right hand side of Eq.~\eqref{eq:cond} has its maximum for $r_s=0.53$ and that it is not possible to satisfy Eq.~\eqref{eq:cond} for $W_S/L_S > 2.08$. Note that $r_s=0.53$ increases the original $W_S/L_S$ ratio of a segment by $1.88$ and the ratio $W_S/L_S$ for the basic scenario with the worst performance ratio is $2.067$, see Section~\ref{sec:basic1}. Therefore, segment $S$ either cannot contain two basic scenarios or its performance ratio cannot exceed $R$, see Lemma~\ref{lem:two_bs}.
\end{proof}

Finally, we allow the introduction of a long job $J_s$ with delay into a WSRPT schedule, that is, the long job does not start at its submission time. For the purpose of continuing with our notation, we temporarily assume that job $J_s$ starts at its submission time and establishes a segment $S$. Note that this assumption violates the condition of an equality instance. Since we use this assumption only for notation purposes, this violation does not matter. In the optimal schedule, job $J_s$ may not start without preemption after all the other jobs of the segment $S$ as shown in Lemma~\ref{lem:non-preempt} for a long job starting at its submission time but it may start earlier or use preemption. 

\begin{lemma}
\label{lem:delayed_jobs}
We assume an introduction of a long job that does not start at its submission time together with additional jobs described by using the segment notation $S$ into a WSRPT schedule. This introduction cannot increase the performance ratio of the WSRPT schedule beyond $R$.  
\end{lemma}

\begin{proof}
Due to Lemma~\ref{lem:group}, the ratio of the start time of job $J_s$ in the WSRPT schedule to the start time of job $J_s$ in the optimal schedule is the performance ratio. Then the WSRPT schedule maintains performance ratio $R$ if the performance ratio of the segment $S$ within the WSRPT schedule is $R$ as well. However, segment $S$ has a smaller performance ratio since job $J_s$ does not start after all jobs in the segment without preemption, see Lemma~\ref{lem:non-preempt}.
\end{proof} 

Formally, we combine our results into a theorem:
\begin{theorem}
\label{thm:WSRPT}
For total weighted completion time scheduling, algorithm WSRPT has a competitive ratio of $R=1.2259$. This competitive ratio is tight for this algorithm.
\end{theorem}

In Fig.~\ref{fig:combined}, we show the already introduced continuous weight profile of the optimal schedule for the existing lower bound of the WSRPT algorithm (top), the basic scenario with the worst performance ratio (middle) and the combination of two basic scenarios generating the competitive ratio of the WSRPT algorithm (bottom).   

\begin{figure}[ht]
\centering
\includegraphics[width=10cm]{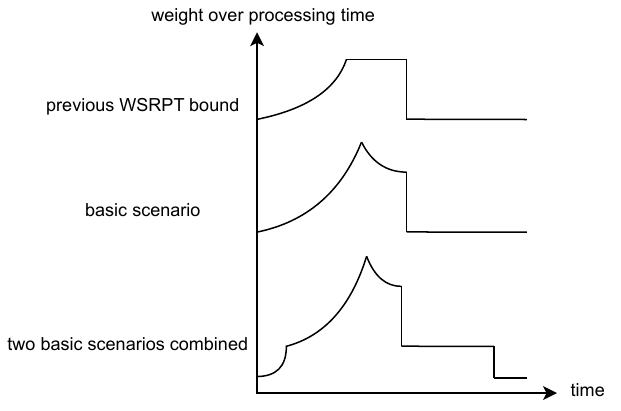}
\caption{\label{fig:combined} Different weight profiles for the lower bound example of \cite{XiC12} (top), the basic scenario with the worst performance ratio (middle), and the example using two basic scenarios and representing the competitive ratio (bottom)}
\end{figure}

\section{Lower Bound}
\label{sec:lb}
	
Although we have significantly improved the known competitive ratio of WSRPT and obtained the best known online algorithm for total weighted completion time scheduling on a single machine that allows preemption, there is still a gap to the best known lower bound $1.0730$ of the competitive ratio. Since the lower bound example of \cite{EpS03} does not use an equality instance, we use a similar approach with equality instances.   

\begin{theorem}
\label{thm:lb}
Every deterministic online algorithm for total weighted completion time scheduling with preemption on a single machine has at least competitive ratio of 1.1038.	\end{theorem}
	
\begin{proof}
The adversary submits two jobs $J_1$ and $J_2$ at time $t=0$ with $p_1=w_1=1$ and $p_2=w_2=2.3364$ and checks the Smith ratios at time $p_1=1$. If $p_2(p_1)=p_2-p_1$ holds then the adversary submits many identical small jobs with processing time $p=\delta$, weight $w= \frac{p_2}{p_2-p_1}\cdot\delta$, and total length $l_1=\sqrt{\frac{2p_2^3-2p_1^3}{p_2}}$ at time $p_1$. Afterwards, the adversary does not submit any further jobs. 

We obtain the competitive ratio
\begin{align*}
c_1 &= \frac{p^2_2 +\frac{p_2}{p_2-p_1}\cdot l_1 \cdot\left(p_2+\frac{l_1}{2}\right)+p_1(p_2+l_1+p_1)}{p_1^2 +\frac{p_2}{p_2-p_1}\cdot l_1 \cdot\left(p_1+\frac{l_1}{2}\right)+ p_2\cdot(p_2+l_1+p_1)}=1 + \frac{p_1\cdot l_1}{p_1^2 +\frac{p_2}{p_2-p_1}\cdot l_1 \cdot\left(p_1+\frac{l_1}{2}\right)+ p_2\cdot(p_2+l_1+p_1)}= 1.1038.
\end{align*} 

If $p_2/p_2(p_1) \geq p_1/p_1(p_1)$ holds then the adversary submits many small jobs with processing time $p=\delta$, weight $w= \frac{p_2}{p_2(p_1)}\cdot \delta$, and total length $l$ at time $p_1$. The adversary chooses $l$ in order to maximize the competitive ratio $c(p_1)$ and does not submit any further jobs and we obtain 
\begin{align*}
c(p_1) &\geq \frac{p_2\cdot (p_1+p_2(p_1)) +\frac{p_2}{p_2(p_1)}\cdot l \cdot\left(p_1+p_2(p_1)+\frac{l}{2}\right)+p_1\cdot(p_2+l_2+p_1)}{p_1^2 +\frac{p_2}{p_2(p_1)}\cdot l \cdot\left(p_1+\frac{l}{2}\right)+ p_2\cdot(p_2+l_2+p_1)}\\
&\geq 1 + \frac{p_1\cdot l+p_2\cdot (p_2(p_1)-(p_2-p_1))}{\frac{p_2}{p_2(p_1)}\cdot l\cdot\left(p_1+\frac{l}{2}\right)+p_2^2+p_2\cdot l_2+p_2\cdot p_1+p_1^2}\\
&\geq 1 + \frac{p_1\cdot l_1+p_2\cdot \left(p_2(p_1)-(p_2-p_1)\right)}{\frac{p_2}{p_2(p_1)}\cdot l_1\cdot\left(p_1+\frac{l_1}{2}\right)+p_2^2+p_2\cdot l_1+p_2\cdot p_1+p_1^2}\\
&\geq 1 + \frac{p_1\cdot l_1}{\frac{p_2}{p_2-p_1}\cdot l_1\cdot \left(p_1+\frac{l_1}{2}\right)+p_2^2+p_2\cdot l_1+p_2\cdot p_1+p_1^2}=c_1.
\end{align*} 

For $p_2/p_2(p_1) < p_1/p_1(p_1)$, the adversary waits until some time $t_s\leq p_2$ with $\frac{p_2}{p_2(t_s)}=\frac{p_1}{p_1(t_s)}$. At time $t_s=p_1+p_2-p_1(t_s)-p_2(t_s)$, the adversary submits many small jobs with processing time $p= \delta$, weight $w=\frac{p_2}{p_2(t_s)}\cdot\delta$, and total length $l$. The adversary chooses $l$ in order to maximize the performance ratio.

If the optimal algorithm starts with $J_1$ then we have
$$\sum w_j\cdot c_j (Opt)=p_1^2 +\frac{p_2}{p_2(t_s)}\cdot l \cdot\left(t_s+\frac{l}{2}\right)+ p_2\cdot(p_1+p_2+l)$$
otherwise we have
$$\sum w_j\cdot c_j (Opt)=p_2^2 +\frac{p_2}{p_2(t_s)}\cdot l \cdot\left(p_2+\frac{l}{2}\right)+p_1\cdot(p_2+l+p_1) .$$
We compare both functions of $t_s$ and obtain
\begin{align*}
&p_1^2 +\frac{p_2}{p_2(t_s)}\cdot l \cdot\left(t_s+\frac{l}{2}\right)+ p_2\cdot(p_1+p_2+l)\geq p_2^2 +\frac{p_2}{p_2(t_s)}\cdot l \cdot\left(p_2+\frac{l}{2}\right)+p_1\cdot(p_2+l+p_1) \; \; \Rightarrow t_s \geq \frac{1}{2}\cdot (p_1+p_2)
\end{align*}
Therefore, the optimal algorithm starts with job $J_1$ for $p_1\leq t_s\leq\frac{1}{2}\cdot(p_1+p_2)$ and we obtain 
\begin{align*}
c(t_s) &\geq \frac{p_2\cdot(t_s+p_2(t_s))+p_1\cdot(p_2+p_1)+\frac{p_2}{p_2(t_s)}\cdot l(t_s)\cdot\left(p_2+p_1+\frac{l(t_s)}{2}\right)}{p_1^2 +\frac{p_2}{p_2(t_s)}\cdot l(t_s) \cdot\left( t_s+\frac{l(t_s)}{2}\right)+ p_2\cdot(p_1+p_2+l(t_s))}\\
&\geq 1+\frac{p_1\cdot \left(p_2-p_2(t_s)+l(t_s)\right)}{p_1^2 + l(t_s) \cdot\left( \frac{p_1\cdot p_2}{p_2(t_s)}-p_1+\frac{p_2^2}{p_2(t_s)}+\frac{p_2\cdot l(t_s)}{2p_2(t_s)}\right)+ p_2\cdot p_1+p_2^2}\\
&\geq 1+\frac{p_1\cdot p_2}{p_1^2 + p_1\cdot p_2-p_1\cdot p_2(t_s)+p_2^2+\frac{p_2\cdot p_2(t_s)}{2} + p_2\cdot p_1+p_2^2} \\
&\geq 1+\frac{p_1\cdot p_2}{p_1^2 + 2p_1\cdot p_2+2p_2^2 +\frac{p_2^2 \cdot \left( \frac{p_2}{2}-p_1 \right) }{p_1+p_2}} = 1.1392 > 1.1038.			
\end{align*}
For the third inequality, we assume $l(t_s)=p_2(t_s)$. Then the performance ratio increases with decreasing $p_2(t_s)$. Therefore, we consider the maximum value of $p_2(t_s)=\frac{p_2^2}{p_1+p_2}$ and obtain the result.

For $\frac{1}{2}\cdot(p_1+p_2)< t_s\leq p_2$, the optimal algorithm starts with job $J_2$ and we obtain
\begin{align*}
c(t_s) &\geq \frac{p_2\cdot(t_s+p_2(t_s))+p_1\cdot(p_2+p_1)+\frac{p_2}{p_2(t_s)}\cdot l(t_s)\cdot\left(p_2+p_1+\frac{l(t_s)}{2}\right)}{p_2^2 +\frac{p_2}{p_2(t_s)}\cdot l(t_s) \cdot\left(p_2+\frac{l(t_s)}{2}\right)+ p_1\cdot(p_2+l(t_s)+p_1)}\\
&\geq 1+\frac{p_1\cdot \left(p_2-p_2(t_s)+\frac{p_2}{p_2(t_s)}\cdot l(t_s)-l(t_s)\right)}{p_2^2 +\frac{p_2}{p_2(t_s)}\cdot l(t_s) \cdot\left(p_2+\frac{l(t_s)}{2}\right)+ p_1\cdot(p_2+l(t_s)+p_1)}\\
&\geq 1+\frac{2\cdot p_1\cdot \left(p_2- p_2(t_s)\right)}{2\cdot p_2^2+p_1\cdot p_2+ p_1^2+ p_2(t_s)\left(\frac{p_2}{2}+p_1\right)}\\ 
&\geq 1+\frac{p_1\cdot p_2}{\frac{9}{4}\cdot p_2^2+\frac{3}{2}\cdot p_1\cdot p_2+ p_1^2}=1.1392 > 1.1038
\end{align*}
Again we assume $l(t_s)=p_2(t_s)$ for the third inequality. Then the performance ratio increases with decreasing $p_2(t_s)$. Therefore, we consider the maximum value of $p_2(t_s)=\frac{p_2}{2}$ and obtain the result.

If there is no such $t_s$ until time $p_2$ then $p_2/p_2(p_2) < p_1/p_1(p_2)$ holds and the adversary submits many small jobs with processing time $p=\delta$, weight $w= \frac{p_2}{p_2(p_2)} \cdot \delta$, and total length $l$ at time $t=p_2$ and stops submission afterwards. We use the notation $l_2=\frac{l_1}{\sqrt{p_2-p_1}}$. The adversary chooses $l$ in order to maximize the performance ratio $c(p_2)$ and does not submit any further jobs and we obtain 
\begin{align*}
c(p_2) &\geq\frac{p_1\cdot (p_1+p_2-p_2(p_2)) +p_2\cdot(p_1+p_2)+\frac{p_2}{p_2(p_2)}\cdot l \cdot\left(p_1+p_2+\frac{l}{2}\right)}{p^2_2 +\frac{p_2}{p_2(p_2)}\cdot l \cdot\left(p_2+\frac{l}{2}\right)+p_1\cdot (p_2+l+p_1)}\\
&\geq 1+\frac{\frac{p_2 \cdot p_1}{p_2(p_2)} l-p_1\cdot l+p_1\cdot(p_2-p_2(p_2))}{\frac{p_2}{p_2(p_2)}\cdot l\cdot \left(p_2+\frac{l}{2}\right)+p_2^2+p_1\cdot p_2+p_1\cdot l+p_1^2}\\
&\geq 1+\frac{\frac{p_2 \cdot p_1}{p_2(p_2)} l_2-p_1\cdot l_2+p_1\cdot(p_2-p_2(p_2))}{\frac{p_2}{p_2(p_2)}\cdot l_2\cdot \left(p_2+\frac{l_2}{2}\right)+p_2^2+p_1\cdot p_2+p_1\cdot l_2+p_1^2}\\
&\geq 1+\frac{p_2\cdot l_2-p_1\cdot l_2}{\frac{p_2}{p_1}\cdot l_2\cdot \left(p_2+\frac{l_2}{2}\right)+p_2^2+p_1\cdot p_2+p_1\cdot l_2+p_1^2}=1.1038
\end{align*}
\end{proof}	

We graphically depict the adversary strategy of the proof of Theorem~\ref{thm:lb} in the game tree of Fig.~\ref{fig:G_Tree}. Depending on the strategy of the online algorithm, there are three different strategies of the adversary.
		
\begin{figure}[ht]
	\centering
			
			\tikzset{
				treenode/.style  = {shape=rectangle, rounded corners, draw, align=center,
					top color=white,bottom color=blue!20},
				Adversary/.style = {treenode, bottom color=red!30},
				Online/.style    = {treenode, bottom color=blue!20},
				PayOff/.style    = {treenode, bottom color=green!10},
				dummy/.style     = {circle, draw}
			}
			\noindent\makebox[\textwidth]{
				\begin{tikzpicture}
				[
				sibling distance        = 15em,
				level distance          = 6em,
				edge from parent/.style = {draw, -latex, thick},
				every node/.style       = {font=\ttfamily\normalsize}
				]
				\node [Adversary] {Adversary}
				child { node [Online] {Online Algorithm}
					child { node [Adversary] {Adversary}
						child[level distance=9em]{node [PayOff] {$C\geq 1.1038$}
							edge from parent
							node[left]{Submits small}
							node[right]{jobs at $t=1$}
						}
						edge from parent node[above, sloped, align=center] {$\frac{p_2}{p_2(p_1)} \geq \frac{p_1}{p_1(p_1)}$}
					}
					child{ node[Adversary]{Adversary}
						child{node [PayOff] {$C\geq 1.1038$}
							edge from parent node[left]{Submits small}
							node[right]{jobs at $t=t_s$}
						}
						edge from parent 
						node[left] {$\frac{p_2}{p_2(t_s)}=$}
						node[right]{$\frac{p_1}{p_1(t_s)}$}
					}
					child { node [Adversary] {Adversary}
						child[level distance=9em]{node [PayOff] {$C\geq 1.1038$}
							edge from parent
							node[left]{Submits small}
							node[right]{jobs at $t=p_2$}
						}
						edge from parent node[above, sloped, align=center] {$\frac{p_2}{p_2(p_2)}<\frac{p_1}{p_1(p_2)}$}
					}    
					edge from parent 
					node[left] {Submits two jobs}
					node[right] {at time $t=0$}  
				};
				\end{tikzpicture}}
			
	\caption{\label{fig:G_Tree}A game tree of the adversary to obtain a high competitive ratio for online total weighted completion time scheduling on a single machine with preemption.} 
\end{figure}
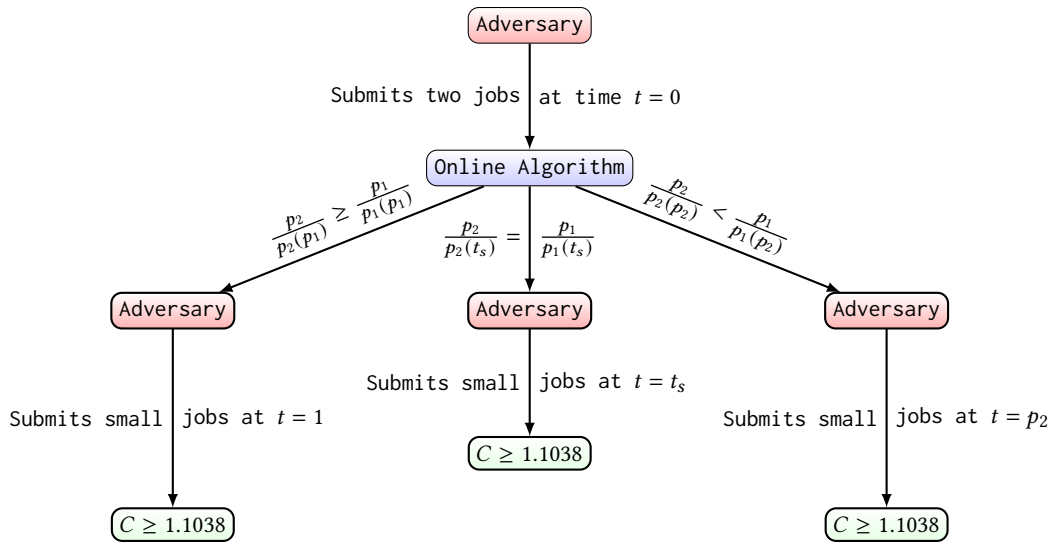  

The example used in the proof of Theorem~\ref{thm:lb} uses two jobs with long processing times and multiple identical jobs with small processing times. Fig.~\ref{fig:LB} shows the performance ratio of the two schedules that complete one long job before starting the other one in dependence of the processing time $p_2$. The general lower bound of the competitive ratio for the problem is $1.1038$ at the point where both curves intersect. In both cases, we select the total length of the small jobs such that we obtain the largest performance ratio. 
		
\begin{figure}[ht]
	\centering
	\includegraphics[width=12cm]{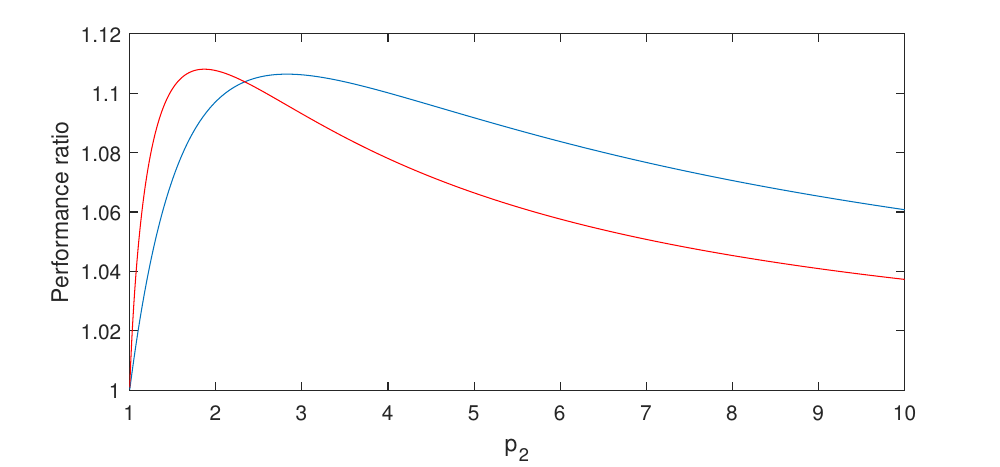}
	\caption{\label{fig:LB}Performance ratios of worst case scenarios over $p_2$ values}
\end{figure}

\newpage

\bibliographystyle{ACM-Reference-Format}
\bibliography{schwiegelshohn}